\documentclass[format=acmsmall, review=true]{acmart}

\usepackage{acm-ec-21}
\usepackage{booktabs} 
\usepackage[ruled]{algorithm2e} 

\SetAlFnt{\small}
\SetAlCapFnt{\small}
\SetAlCapNameFnt{\small}
\SetAlCapHSkip{0pt}
\IncMargin{-\parindent}

\setcitestyle{acmnumeric}

\usepackage{graphicx}

\usepackage{amssymb}

\usepackage{amsmath}

\usepackage{amsthm}
\usepackage{color, colortbl, xcolor}
\usepackage{appendix}

\newtheorem{theorem}{Theorem}
\newtheorem*{theorem*}{Theorem}
\newtheorem{claim}{Claim}
\newtheorem*{claim*}{Claim}

\newtheorem{definition}{Definition}

\newcommand{\objectives}{\mathcal{F}}
\newcommand{\system}{\mathcal{S}}

\newcommand{\eval}{Eval}
\newcommand{\verify}{Verify}
\newcommand{\setup}{Setup}

\newcommand{\comp}[1]{C_{#1}}
\newcommand{\boldheader}[1]{\vskip 5pt \noindent{\bf #1}}

\begin{document}

\title{Proofs of Useless Work \\ Positive and Negative Results for Wasteless Mining Systems}

\author{Maya Dotan}
\email{MayaDotan@mail.huji.ac.il}

\author{Saar Tochner}
\email{saart@cs.huji.ac.il}
\affiliation{%
  \institution{The Hebrew University of Jerusalem}
  \country{Israel}
  }

\begin{abstract}
Many blockchain systems today, including Bitcoin, rely on Proof of Work (PoW). Proof of work is crucial to the liveness and security of cryptocurrencies. The assumption when using PoW is that a lot of trial and error is required on average before a valid block is generated. One of the main concerns raised with regard to this kind of system is the inherent need to ``waste'' energy on ``meaningless'' problems. In fact, the Bitcoin system is believed to consume more electricity than several small countries.

In this work we formally define three properties that are necessary for wasteless PoW systems: (1) solve ``meaningful" problems (2) solve them efficiently and (3) be secure against double-spend attacks. These properties aim to create an open market for problem-solving, in which miners produce solutions to problems in the most efficient way (wasteless). The security of the system stems from the economical incentive created by the demand for solutions to these problems.

We analyze these properties, and deduce constraints that must apply to such PoW systems. In our main result, we conclude that under realistic assumptions, the set of allowed problems must be preimage resistant functions in order to keep the system secure and efficient.

\end{abstract}

\begin{titlepage}
    \maketitle
\end{titlepage}

\section{Introduction}
Cryptocurrencies (such as Bitcoin~\cite{nakamoto2008bitcoin}) are distributed (and often decentralized) currencies. Bitcoin operates on top of the Blockchain in which each block encapsulates monetary transactions. 
A transaction is valid only upon being included in a block. Security in Bitcoin translates to ensuring that the Blockchain is constantly appended, and it is appended in the same way across all users in the system (consistency). Appending the blockchain is done through a process called "Mining", and block creators are called "Miners". 
It is of vital importance to the health of the protocol that the rate of blocks created is regulated, and that it is not controlled by an adversary. 
The most popular method for regulating block creation is through "Proof of Work", where miners must perform a sufficient amount of "computational work" (e.g. solving a cryptographic puzzle) in order to create a block. This implies that with high probability, miners can only create blocks at a rate which is proportional to their computational power in the network. From this property stems the security guarantee of Bitcoin - As long as no single user controls a majority of the computational power in the network, then the probability of inconsistency across users decreases exponentially with the number of blocks created. In this sense, proof of work is what ensures that the Bitcoin system is secure. 

The mining process introduces a serious environmental problem due to its massive energy consumption. The energy consumption of Bitcoin is estimated to be at least as high as that of some small countries \cite{o2014bitcoin,de2018bitcoin}.
In this paper we argue that a computation is not wasteless if someone is willing to pay for its solutions in some external setting. Papers such as \cite{zhang2017rem,zheng2020axechain,oliver2017proposal} followed the same approach, and designed systems that enable users to request problems that they need to solve, and change the mining process to solve these problems. The main contribution of this paper is the modeling and analysis of the economical market that is created by such mining systems. We model users that upload problems as the consumers -- they ask for solutions to their computational problems and are willing to pay a fee. Miners are the producers of goods. They invest energy in order to produce solutions to problems and collect a fee. Their profit is a combination of this fee and the block reward. we connect between the environmental problem and economic waste. We show that mining systems in which not all energy goes to producing solutions to problems are wasteful also in the profit margin of the producers. We therefore limit our discussion to systems that are ``energy efficient'' which we will define formally and analyze in the main part of this work.

Finally, any wasteless mining system must still remain secure according to the standard notions of security in the Blockchain world today.
Our economical modeling exposes a new challenge in this regard that has not been addressed in previous works: How to incentivize miners to invest their computational power to solve users problems inside the mining system instead of in any other external settings. 
This is crucial to the liveness and security of the system since blockchain systems' security increases with the volume of work that goes into the proof of work process. If miners prefer to invest their computational power in an external setting the security of the system becomes compromised.

Combining all of the above we have that our discussion of ``useful work'' should be limited to mining systems that uphold all three properties (1) Meaningful (User Uploaded), (2) Energy-Efficient and (3) Secure.

\subsection{Related Work}
"Proof of Stake" and "Proof of Space", which studied in \cite{bentov2014proof}, \cite{gilad2017algorand}, \cite{dziembowski2015proofs}, \cite{kiayias2017ouroboros} and more, replace the energy with a different resource. While these avoid energy waste, they incur waste in other domains. We therefore consider these approaches to be only as partial solutions.

We focus on an approach first presented in \cite{king2013primecoin}, and again in \cite{ball2017proofs}. They introduce the notion of "Proof of Useful Work". In these systems, the outputs of the "mining computation" are supposed to be meaningful. Both of these works however do not allow users to upload their own problems (the problems are dictated by the system). We claim that in order to make such systems favorable, users must be allowed to upload computational tasks that have some value to them. This creates a competitive market, which provides incentives for miners to participate in the solution of computational problems for a profit.

A step forward in implementing useful proofs of work with user-uploaded problem has been done in REM~\cite{zhang2017rem}, however it is strongly based on 2 facts: (i) The hardware enforces correct reporting of work, and (ii) The assumption that all miners use this specific hardware. This can be viewed as a special case of the general solution we describe, where a trusted setup can verify the complexity on the computation. We elaborate on this in appendix~\ref{app::REMspecialCase}.

More papers were published in this field, however none of them meet all three properties. These include:  \cite{oliver2017proposal}, \cite{zheng2020axechain} miners' contradicting incentives can cause waste, \cite{daian2017short} solve the case for restricted, non-user-uploaded problems and \cite{ball2017proofs}.




\subsection{Our Contribution}

In this paper we formally define the notion of "wasted energy" - energy is "wasted" if no one is willing to pay for the result of the computation. Thus, a ``meaningful'' problem is one that a user in the system is willing to pay for.

We look at the trade-off between solving meaningful problems, reducing marginal computation work, and keeping the system secure.
Formally, we define three desired properties for a ``non-wasteful'' proof of work system:
(1) Meaningful Problems - The results of computations performed by miners should be of interest. Interest is measured by economic incentive - the user must be willing to pay for the result of the computation. A simple economic argument shows that this reduces to "User Uploaded" problems. 
(2) Energy Efficiency - The algorithms used to solve the problems are optimal. This again makes sure that there is no ``waste'' in the mining process as compared to an external setting.
(3) Security - the system should by secure against double spend attacks by a minority attacker with overwhelming probability. In particular our definition coincides with the common prefix property as defined in \cite{garay2015bitcoin}


\begin{figure}
	\centering
	\includegraphics[scale=0.30]{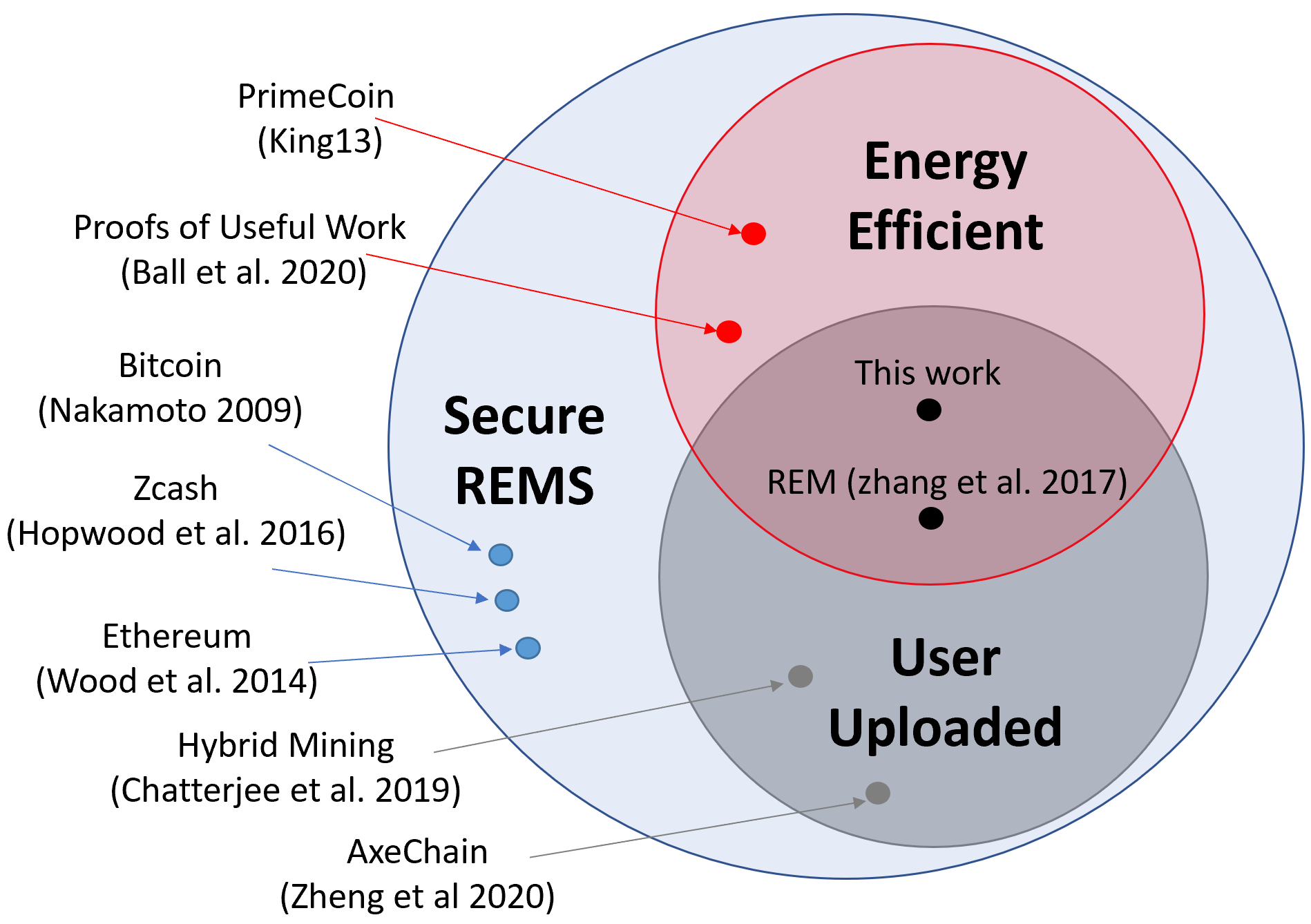}
	\caption{Intersections of the three requirements compared to state of the art systems today} 
\end{figure}

In high level, We prove the following necessary conditions:
(i) Solving user-uploaded problems must be an integral and mandatory part of the mining process. 
(ii) Miners must supply a proof that they attempted to solve user-uploaded problems. These proofs need to be easy to verify and "unfakeable" given any prior knowledge on the problems or the solutions. 
(iii) Anyone with computational resources should be economically incentivized to solve a problem within the mining process rather than offline. 
(iv) The system should be able to evaluate the computational requirements for solving a mining objective. 
(v) Following the previous point, unless there is trusted reporting of the amount of computational work done by a user, the system should only allow users to upload problems for which every case is the "worse case" (For example, inverting one-way functions).

\subsection{Structure}
In Section~\ref{sec::pricingObjectives} we relate the properties of meaningfulness and energy-efficiency to market powers and economics.
In section~\ref{sec::Model} we formally define the model, section~\ref{sec::modelAnalysis} formal analyzes necessary properties in systems that has our three properties. We fully characterize the family of possible mining systems. 
In section~\ref{sec::conclusions} we discuss some future directions. 
In addition, in the appendix there is an example implementation of a mining system which meets all $3$ criteria, discussion on current solutions, and all the omitted proofs.

\section{Mining as a Free Market For Problem Solving}\label{sec::pricingObjectives}

Throughout this paper we discuss the possibility of users uploading computational puzzles to which they need solution, alongside a fee. A miner who solves the mining objective collects the fee. 
We would like to discuss why a fee is in fact necessary for the liveness of the system. 
Assume for a moment an REMS that is (1) Secure (2) Energy Efficient and (3) User Uploaded, and users upload mining objectives without a fee (or with a negligible fee). We claim that if the amount $\textrm{Block Reward} + \textrm{mining objective Fees} + \textrm{Transaction Fees}$ is less than the amount $\textrm{Energy Needed to Mine} + \textrm{Value of solutions in outside market}$, miners will not be incentivised to mine, and would prefer to solve mining objectives in an offline setting. This is of course harmful to the security of the system (since less miners implies less security), and to the liveness of the system. It is therefore essential to make sure that introducing the option to upload mining objectives to a mining system does makes economic sense - it should be financially appealing to upload mining objectives for a fee.
We argue that users will in fact be incentivised to solve their mining objectives through the REMS. This is because the price for a solution within the REMS is lower the any outer market. This is the result of the REMS miners that are incentivised by block rewards and transaction fees in addition to the fee for solving mining objectives.

The security and energy-efficiency requirements also have to do with market powers.
Just like a factory would aim to minimize its costs by only paying the essential minimal for raw material and operational cost in order to be profit maximizing, in any REMS that makes economical sense, the price in terms of resources to produce a solution to a mining objective should be minimized. If this is not the case, it would be cheaper to solve the problem in an exterior setting, which is detrimental to the REMS. This in turn means that the algorithms for solving mining objectives should be optimal, and that the overhead spent by the system on anything other than solving mining objectives should be minimized (``energy-efficiency'').

This connection between security and economic efficiency in system design exists in many aspects of cryptocurrency systems. There have been several attacks on systems such as Bitcoin using this type of economical reasoning, One famous attack is the Selfish Mining attack \cite{eyal2014majority} which leads to a waste of energy caused by faulty behaviour of profit-maximizing miners. These however are not usually framed in the economical setting.

\section{Formal Model for Mining Systems}\label{sec::Model}
In this Section we formalize the properties that any wasteless mining system must uphold. 

\subsection{Definitions}
Throughout this section we will discuss strings in $\{0,1\}^{256}$, but the same would hold for $\{0,1\}^D$ for any $D\in\mathbb{N}$.

\begin{definition}[Mining
Objective]\label{def::MiningObjective}
    Let $D\in\mathbb{N}$.
    A \textit{mining objective} is a pair $\left<f,y\right>$ where $f$ is a function and $y\in Im(f)$. We will say that a mining objective $\left<f,y\right>$ was solved if some $x \in \{0,1\}^{D}$ was found such that $f(x) = y$. 
\end{definition}

We next define the entities in the system:

\begin{definition}[Users and Miners]
    There are two types of users in the system:
    \begin{enumerate}
    \item Users - This is the basic player in the system. They can upload mining objectives to $\objectives$ for a fee, and participate in the verification process - they run $\verify$
    \item Miners -  These are users that in addition to the above also participate in block creation (run $\eval$).
    \end{enumerate}
    
\end{definition}

\begin{definition}[Preimage Resistance Function]\label{def::preimageResistanceFunction}
    A preimage resistance function is a function that is easy to compute, but whose inverse is difficult to compute. 
    More precisely, a preimage resistance function $f$ holds the property that for every adversary algorithm $\mathcal{A}$ that runs in polynomial time in $size(x)$, $Pr[f(\mathcal{A}(f(x))) = f(x)]$ is negligible. 
     \cite{lamport1979constructing,rogaway2004cryptographic}
\end{definition}

Notice that this means that miners may upload mining objectives to $\objectives$. In particular, they may upload problems to which they already know the solution, or mining objectives that they have an advantage in solving.

\begin{definition} [REMS] \label{def::REMS}
A \textbf{Repeated Eval Mining System} $\system = ( \objectives,\setup, \eval, \verify)$ is a quadruple defined as follows:
\begin{itemize}
    \item $\objectives$ is a set of mining objectives.
    \item $\setup() \rightarrow ek$ is a randomized, polynomial time algorithm that takes no parameters, and returns an evaluation key $ek$, which is unpredictable to all users. 
    \item $\eval(\objectives, ek, x) \rightarrow (y, \pi)$ is a polynomial time algorithm which takes an input $x \in \{0, 1\}^*$ and produces an output $y \in Y$. Where $Y=\{\mathbb{T}, \mathbb{F} \} \times \{\mathbb{T}, \mathbb{F}\}^{\objectives}$. The first coordinate of $y\in Y$ is an indication of whether the seed $x$ results in a successful mine (a block was created), and the rest is an indicator vector of which mining objectives were solved by $x$. $\pi \in \{0, 1\}^*$ is a (possibly empty) proof $\pi$. 
    \item $\verify(\objectives, ek, x, y, \pi)\rightarrow \{ \mathbb{T}, \mathbb{F} \}$ is a deterministic polynomial time algorithm which returns $\mathbb{T}$ if $y,\pi$ is a valid output of $\eval(\objectives, ek, x)$, and $\mathbb{F}$ otherwise. 
\end{itemize}

In an REMS, repeated calls to $\eval$ are made, and when a query returns $y[0]=\mathbb{T}$ then a mining attempt was successful, and we will refer to this as a \textbf{new block was mined}.
\end{definition}

For example, in today's Bitcoin, $ek$ is the hash of the previous block, $x$ is some string, and there exists a parameter $D_B$, such that $\objectives = \{\left< \text{SHA}256\left(ek\circ x\right), y \right>| \forall y \text{ s.t } y\le D_B\}$. $\eval(\objectives, ek,x)$ executes SHA$256$ on $ek \circ x$,  gets an output $\hat{y}$, and returns the vector $y$ and proof $\pi = x$ where: $\forall \left<f,y\right> \in \objectives$ it holds that $y_{\left<f,y\right>}=\mathbb{T}$ i.f.f $\hat{y} == y$, and $y[0]=\mathbb{T}$ i.f.f. one of the objectives was solved.
$\verify$ returns $\mathbb{T}$ i.f.f the proof (in this case $\pi = x$) maintains that $\text{SHA}256(ek\circ \pi) \le D_B$.

\vspace{\baselineskip}
Necessary conditions that should hold for any REMS system are (1) Correctness - $\verify$ accepts proofs that were generated by $\eval$ and (2) Soundness -  the probability to generate such that $\verify$ will accept is negligible (regardless off the energy efficiency question). Formally:

\begin{definition}[Correctness]\label{def::correctness}
    We will say that an REMS $\system = (\objectives,\setup, \eval, \verify)$ is correct if $\eval(\objectives, ek,x) = (y, \pi)$ $\Longrightarrow verify(ek, x, y, \pi) = \mathbb{T}$. 
\end{definition}

\begin{definition}[Soundness]\label{def::soundeness}
\label{definition::soundness}
    We will say that an REMS $\system = ( \objectives,\setup, \eval, \verify)$ is sound if for every algorithm $A$ that runs in time $O(poly(ek))$,
    and $ek$ which is sampled uniformly from $\{0,1\}^{D}$: If the output of $A(ek)$ is $(x_{ek}, y_{ek}, \pi_{ek})$ then
    \begin{align*}
       Pr [ verify(ek, x_{ek}, y_{ek}, \pi_{ek}) = \mathbb{T} 
                 \wedge (y_{ek}, \pi_{ek}) \neq \eval(ek,x_{ek}) ] \le 2^{-|ek|}
    \end{align*}
    
\end{definition}

From now on, we will only discuss REMSs which are both Sound and Correct.

Throughout the following sections, for the sake of ease of notations, assume that for every $x$ the task of evaluating $f(x)$ takes the same amount of computational resources for every $f\in \objectives$. 
Moreover, we assume from now on that $\eval$ consumes a constant amount of computational  resources per execution. We will address the general case and discuss why these assumptions is not necessary in Section~\ref{sec::conclusions}.

\begin{definition}[Computational Resources Demands]\label{def::Df_Cf}
    Let $\system = ( \objectives,\setup, \eval, \verify)$ be an REMS. Let $\left<f,y\right>$ be a mining objective in $\objectives$.

    \begin{enumerate}
        \item We denote $D_f$ the computational resources necessary to compute the value $f(x)$ for any $x\in\{0,1\}^{D}$, when using optimal algorithm for computing $f(x)$.
        \item Denote $\comp{\left<f,y\right>}$ the computational resources necessary, in expectation, to find a solution $x\in\{0,1\}^{D}$ such that $f\left(x\right)=y$, where $x$ is sampled uniformly from $\{0,1\}^{D}$. Again, the computation assumed the optimal algorithm for computing $f(x)$.
    \end{enumerate}
\end{definition}

Finally, we define the blockchain data-structure over an REMS $\system$:
\begin{definition}[Blockchain] \label{def::blockchain}
    Let $\system = ( \objectives,\setup, \eval, \verify)$ be an REMS.
    A linked list of blocks $B_1, \cdots, B_n$ is called a blockchain if for every block $B = (ek_B, x_B, y_B, \pi_B)$ in the blockchain, it holds that $\verify(\objectives, ek_B, x_B, y_B, \pi_B) = \mathbb{T}$. Moreover denote: $B_1$ is called the genesis block, $B_n$ is the head, and $n$ is the weight of the blockchain.
\end{definition}

Note that in the above definition, we assume that every block for which $\verify$ returns $\mathbb{T}$ is a block with the same weight as all other blocks. This notion can in theory be generalized, however we will not go into this case in this work.

\subsection{Security, Energy Efficiency, Meaningfulness}
We first define security in REMSs. Our definitions use notations inspired by  \cite{boneh2018verifiable}, and the definitions coincide with \cite{garay2015bitcoin}. Our definitions of security still are with respect to common prefix property and the chain quality property, and require an honest majority. 

\boldheader{Threat Model}
The attacker is assumed to be able to produce as many identities as desired and can alter the code executed by the users under its control. In addition we assume that any user may be a miner (active in the voting procedure).
We assume that the attacker can first observe the activity of the honest users and be the last user to decide on a strategy. We assume that the attacker does not control a majority of the computational power in the network, that it is computationally bounded and that regular cryptographic assumptions hold.

The attacker goal is to mine more blocks than its relative computational power. To do so, it may upload as many mining objectives as it wants and may attempt to manipulate the choice of problems that it and other miners are trying to solve. It can not modify the logic that is executed be running $\eval, \verify$ by other users. 


\begin{definition}[Secure REMS] \label{def::Security} 

Let $\system = ( \objectives,\setup, \eval, \verify)$ be an REMS. 
Let $m$ be a miner in the system.
We say that the REMS $\system$  is \textbf{secure} if it holds that:
\begin{equation*}
    \begin{aligned}
        Pr \left[m\textrm{ finds some } x \textrm{ such that } \eval( \objectives, ek, x)_0 = \mathbb{T}  \right]= \left[\frac{\textrm{\# Executions of $\eval$ from $m$}}{\textrm{\# Executions of $\eval$ across the network}}\right]
    \end{aligned}
\end{equation*}
Where $m$ gets to choose the distribution over $\{0,1\}^{D}$ from which she samples $x$ (without knowing $ek$).
\end{definition}

The guarantee is that an attacker can not create blocks faster than its ratio of the total computational power in the network. This means that the "cryptopuzzle" should uphold the property that $$Pr\left(\textrm{User } m \textrm{ solves the puzzle}\right) = \frac{\textrm{Computational power of }m}{\textrm{Total computational power in the network}}$$
If this property does not hold we say that the system is vulnerable to double spend attacks by a minority attacker.

Note that this notion of security coincides with the notion of security in ``The Bitcoin Backbone Protocol" \cite{garay2015bitcoin}. As shown there, this is enough to ensure that the common prefix property and the chain quality property are maintained in the system.

\begin{definition}[Meaningful REMS]\label{def::meaningfulREMS}
    We say that a mining objective is meaningful if there exists a user willing to pay for the resources that are required to solve it regardless of the mining process (i.e. the user would also pay for a solution in an external setting). We say that an REMS $\system = ( \objectives,\setup, \eval, \verify)$ is meaningful if all of the mining objectives in $\objectives$ are meaningful. 
\end{definition}

Due to to Definition~\ref{def::meaningfulREMS} we will from now on use the terms ``meaningful'' and ``user-uploaded'' interchangeably.

\begin{definition}[Secure-User-Uploaded REMS]\label{def::UserUploadedSecurity}
Let $[M]$ be the set of all the users in a system (including all miners) . Let $\objectives$ be a set of mining objectives that was chosen by $[M]$ \footnote{In particular, a miner $m\in[M]$ may have any non trivial amount of information about any mining objective $\left<f,y\right>\in\objectives$, such as a solution $x$ for which $f(x)=y$.}.
We say that it is ``secure user-uploaded'' if $\system = ( \objectives,\setup, \eval, \verify)$ secure.
\end{definition}

Definition \ref{def::UserUploadedSecurity} is the formalization of the combination of the two conditions discussed in the introduction. Notice that the requirement that the mining objective was chosen by the miner (and not a user that is participating only through uploading problems) is necessary in order to ensure security. This is because we want any secure REMS to be resilient to miners maliciously uploading mining objectives in order to increase their chances of successfully mining a block. 

For the next part of our definitions we use the following intuition:
We say that a mining objective $\left< f, y \right> \in \objectives$ is meaningful if there is a user who is willing to pay for the computational resources that are needed in order to solve it. We would like to make sure that any system that meets our requirements will only allow for meaningful objectives to belong to $\objectives$. Since this is not yet well defined, we begin with the following softer definition of energy efficiency. Combining this notion with the fact that $\objectives$ is composed of user uploaded mining objectives, we can describe necessary conditions for systems in which all mining objectives are meaningful. In particular, an energy-efficient REMS that operates over user-uploaded mining objectives could be described as a market, with the users as consumers (consuming the solutions to the mining objectives) and miners as producers (invest energy in order to fulfill the consumers' demands).

\begin{definition}\label{def::EnergyEfficientREMS}[$\epsilon$-Energy Efficient REMS] 

Let $\epsilon>0$. We say that a mining system $\system = ( \objectives, \setup, \eval, \verify)$ is $\epsilon$-energy efficient if it holds that for every $x\in\{0,1\}^{D}$ the energy ratio: $$\frac{\Sigma_{\left<f,y\right>\in\objectives} D_f}{\textrm{computing } \eval\left(\objectives,ek,x\right) } > 1-\epsilon$$
That is, the system can make sure that the percent of energy used for solving user uploaded mining objectives is arbitrarily close to $1$.
\end{definition}

From now on, any time we say that a claim holds for an ``energy-efficient'' REMS, we mean that it holds for an $\epsilon$-energy-efficient REMS \textbf{for every $\epsilon\ge 0$}.

\section{Necessary Properties of REMS - Formal Analysis}\label{sec::modelAnalysis}
In this section we prove our main theorems about Repeated Evaluation Mining systems. We fully characterize the allowed set of functions that may belong to $\objectives$ that meet all 3 of the desired criteria: Secure, Energy Efficient and Meaningful.

The claims are organized in the following structure. First, we discuss secure REMSs and show two basic properties which we prove are necessary for any REMS to be secure: The function $\eval$ should be optimally efficient and the relationship between mining a block and solving a mining objective should be correlative to ``how hard" the objective is. 

Next, we refine the discussion to secure \& energy-efficient REMS, wherein we introduce the proof $\pi$ to our analysis. We initially show that ``hard-to-generate" proofs are mandatory in order to keep the system energy-efficient. Using the above basic properties in our setting we prove (i) $\eval$ should be optimally efficient in generating the proof, and (ii) the relationship between finding a coherent proof and solving a mining objective should be, as above, correlative to the resources demand of solving the objective.

In the last subsection we present our main theorems, which hold for secure, energy-efficient and user uploaded REMSs. We combine the claims, and deduce constraints on the allowed set of mining objectives in such systems.

\subsection{Secure REMS}\label{subsec::security}

\begin{claim}[$\eval$ is optimal]\label{claim::evalOptimal}
Let $\mathcal{S}=( \objectives, \setup, \eval, \verify)$ be a secure REMS.
Then there does not exist any algorithm $\eval'$ that is more efficient than $\eval$ such that $\mathcal{S'}=( \objectives, \setup, \eval', \verify)$ is sound. \footnote{This condition means that $\eval$ is the optimal algorithm for the computational task which mining is based on.}
\end{claim}
\begin{proof}
    Assume towards a contradiction that there exist $\eval'\neq \eval$ which is more efficient than $\eval$. Assume w.l.o,g that the execution of $\eval'$ is more efficient than that of $\eval$ by a factor of $\alpha> 1$
    Assume that an attacker uses $\eval'$ instead of $\eval$, while all other users use $\eval$.
    Then it holds that the portion of block awarded to the attacker in expectation is:
    
    $$\frac{\textrm{\# executions of attacker using }\eval'}{\textrm{\# of executions of }\eval + \textrm{\# executions of }\eval' } =$$
    $$\frac{\alpha\cdot\textrm{\# executions of attacker using }\eval}{\textrm{\# of executions of }\eval + \textrm{\# executions of }\eval' } >$$
    $$\frac{\textrm{\# executions of attacker if they used  }\eval}{\textrm{\# of executions of } \eval \textrm{ if everyone used }\eval}$$
    Which is an honest miners' probability of mining a block.
    So an attacker increases the speed at which it mines a block as compared to the honest network, which is a contradiction to the notion of security defined in~\ref{def::Security}.
\end{proof}


In the following claim, we formalize the following notion: 
In a secure REMS, solving each mining objective will result in successfully mining with a probability that is proportional to the resources demand of the objective.

\begin{claim}\label{claim::weightsOfMiningObjectives}
For any $\left<f',y'\right>\in\objectives$ it holds that:
$$Pr\left(x \textrm{ results in a block }| f(x) = y\right)  = Pr\left(x \textrm{ results in a block }| f'(x) = y'\right)\cdot \frac{\comp{\left<f,y\right>}}{\comp{\left<f',y'\right>}}$$

Where $\comp{\left<f,y\right>}$ is the expected amount of computational power required for Finding a solution $x$ that satisfies $\left<f,y\right>$.
\end{claim}
 
The proof for Claim~\ref{claim::weightsOfMiningObjectives} appears in Appendix~\ref{app::Security}. The proof's idea is that an attacker with prior knowledge on the ``easier" mining objectives will tend to focus on solving them, unlike the honest miners. This will give the attacker an unfair advantage which contradicts the security of $\mathcal{S}$.


Restating Claim~\ref{claim::weightsOfMiningObjectives} in different words results in the very harsh requirement that $\comp{ \left< f, y \right> }$ must be known to the system (or, at least, relative to all other mining objectives in the system).

If this were not true, the system would have to be able to assess the amounts $\comp{ \left< f, y \right> }$ for every mining objective $\left<f,y\right>$ that is uploaded to the system by a user. This, in general, is a computationally infeasible (undecidable) task. 


\subsection{Secure \& Energy-Efficient REMS} \label{subsec::energyEfficient}

We now turn our attention to analyzing the energy efficiency requirement. Namely we show that the energy efficiency property implies that the system should not use additional energy resources for anything other than the task of solving mining objectives, up to a negligible amount (dictated by $\epsilon$). 

In the following claim we introduce proofs (of work) as a tool for enforcing that miners indeed solve mining objectives. This is important since it formally captures the following intuition: Miners must always be incentivised to solve the mining objective, rather than generating proofs in some way that is external to the system.

\begin{claim}[Proofs are necessary]\label{claim::energyEfficientProofsNecessarry}
    Let $\mathcal{S}=( \objectives, \setup, \eval, \verify)$ be an energy efficient REMS. 
    Then miners must supply proofs of attempting to solve mining objectives from $\objectives$ as part of the mining process. 
    Furthermore, the computational difficulty of computing $\eval$ must be less than the difficulty of 
    finding a proof $\pi$ for which $\verify$ evaluates to $\mathbb{T}$
\end{claim}

This means that if $\pi$ was omitted from the definition of REMS, then the system could never hold both conditions of security and energy efficiency. Thus the requirement for proofs is actually a harsh requirement of any secure and energy efficient REMS.
The key property is that given $ek, x, y \in [0,1]^*$, it is hard to find a proof such that $\verify(ek, x, y, \pi) = \mathbb{T}$. I.e. the probability of success of any algorithm $ALG$ that tries to find $\pi \in [0,1]^*$ such that $\verify(ek, x, y, \pi) = \mathbb{T}$ is extremely low (following Definition of soundness).

The proof of Claim~\ref{claim::energyEfficientProofsNecessarry} appears in Appendix~\ref{app::energyEfficiency}, as it is quite technical and long.

For example, in today's Bitcoin, we can view the ``nonce" that is attached to the block header as the embedded proof. Under the terminology of our paper, the propose of this nonce is to prove that the miner searched for a solution to Bitcoin's mining objectives (which are discussed after Definition~\ref{def::REMS}). Note that in Bitcoin, in addition to the proof, the miner also provides the solution to the mining objective -- the block header.

\begin{claim} \label{claim::evalOptimalProver}
    Let $\mathcal{S}=( \objectives, \setup, \eval, \verify)$ be a secure and energy efficient REMS. Then $\eval$ is the optimal algorithm for generating the proof $\pi$ given setup $ek$.
\end{claim}

The proof for Claim~\ref{claim::evalOptimalProver} appears in Appendix~\ref{app::energyEfficiency}. Note that this proof is very similar to the proof of Claim~\ref{claim::evalOptimal}.

In claim~\ref{claim::evalOptimal} We proved that $\eval$ must execute the optimal algorithm for solving mining objectives in $\objectives$. In claim~\ref{claim::evalOptimalProver} we showed that $\eval$ is the optimal algorithm for generating proofs at attempting to solve mining objectives from $\objectives$.
We point out that in the special case that the proof of trying to solve a mining objective $\left<f,y\right>$ using input $x$, is exactly the output $f(x)$, then the two claims are identical. However, in the general setting, this need not be the case.
We further discuss this special case of the output being the proof in Theorem~\ref{Theorem::PreImage}.

The following Claim~\ref{claim::findProofMeansSolveObjecives} formalizes the following trait: If successfully mining a block depends on solving a mining objective (i.e. block creation is the result of finding a correct solution to a user uploaded question), then a secure REMS should not allow an attacker to mine more blocks (than their proportional computational resources) by solving ``easier" problems. This claim presents the tradeoff between the computational power that is needed in order to solve a mining objective and the probability to successfully mine a block.

\begin{claim} \label{claim::findProofMeansSolveObjecives}
     
     Let $\mathcal{S}=( \objectives, \setup, \eval, \verify)$ be a secure-energy-efficient REMS. 
     
     Let $\objectives = \{\left<f_1,z_1\right>,\ldots,\left<f_n,z_n\right>\}$
    be the mining objectives. Then for any $\left<f_i,y_i\right>,\left<f_j,y_j\right>\in\objectives$ and for all $x=x_1x_2\ldots x_n \in\{0,1\}^{D}$, for any algorithm used to generate proofs it holds that:
    \begin{align*}
        &Pr\bigg(\textrm{Find $y,\pi$ s.t.: } \verify(\objectives, ek,x,y,\pi) =\mathbb{T} |
        \eval(\objectives,ek,x)_{\big<f_i, z_i\big>} = \mathbb{T}\bigg)  = 
    \\
        \frac{\comp{\left<f_i,z_i\right>}}{\comp{\left<f_j,z_j\right>}} \cdot &Pr\bigg(\textrm{Find $y,\pi$ s.t.: } \verify(\objectives, ek,x,y,\pi) = \mathbb{T}|
        \eval(\objectives,ek,x)_{\big<f_j, z_j\big>} = \mathbb{T}\bigg)
    \end{align*}
    Where $y$ is the output of $\eval(\objectives,ek,x)$.
\end{claim}

The proof for Claim~\ref{claim::findProofMeansSolveObjecives} appears in Appendix~\ref{app::energyEfficiency}.

From the above we conclude that the objectives $\objectives$ can only contain mining objectives for which it is hard to generated pairs $(x,\pi)$ for which $\verify(\objectives, ek,x,y',\pi)$ will be evaluated to true. So in order for an objective $\left<f,y\right>$ to be legal, it should be both (1) equally hard to solve across all users and (2) equally hard to generate a proof for all users.
    
For example, assume that a system can be designed through utilizing mining objectives that are $3$-SAT problems, and the proofs are possible assignments (i.e. $\left<\Phi,y\right>$ where $\Phi$ is a $3$-SAT formula, $y = \mathbb{T}$ and the proof $\pi$  is a binary string symbolizing which clauses in $\Phi$ are $\mathbb{T}\text{ or }\mathbb{F}$). Assume in addition that an adversary miner has the following non-trivial information about a mining objective $\Phi$: the adversary knows that for any assignment $x$ it holds that in $\Phi(x)$ at most half of the clauses are satisfiable. Then this miner can avoid verifying the assignment to every clause in $\Phi$ if it discovers that half of the clauses have already been satisfied, reducing the amount of computations it has to use. This way the miner increases their relative power in the system, which contradicts security.\footnote{An interesting note is that the ``efficiently verifiable'' requirement of mining objectives in $\objectives$ implies that $\forall \left<f,s\right> \in \mathcal{F}$, it holds that $f$ is in $NP$.}

To conclude the discussion, we have seen that requiring miners to supply proofs of attempts of solving mining objectives from $\objectives$ is necessary, and these proofs must be hard to fake -- for a mining objective $\left<f,y\right>$ it should be at least as hard as finding $x$ such that $f(x)=y$; and this is exactly the amount $C_{\left<f,y\right>}$.


\subsection{Secure \& Energy-Efficient \& User-Uploaded REMS}

We now turn to discuss the case of ``user-uploaded'' REMSs. 
Let $\left<f,y\right>$ be a mining objective in $\system$. Intuitively, the system should know how to estimate the relative resources that it takes to find a solution to any other mining objective $\left<f,y'\right>$ compared to $\left<f,y\right>$. In addition, the system should not spend a lot of energy in computing this information, since it has to meet the harsh requirement of energy efficiency.
This means that the system must have information on the amounts $D_f$ and $C_{\left<f,z\right>}$ for every $z\in \{0,1\}^{D}$. 

For example, an REMS $\system$ that allows all functions of the form $f = (SHA256,conf)$ for $conf \in \mathbb{N}$ where $f(x) = SUB\_STRING(SHA256(x),from=0,to=conf)$ can meet the above requirement. 

An example of a system that can not meet this property is some $\system'$ with mining objectives $\left<f,y\right>$ where $f$ is a general SAT problems (or any other NP-hard problem); Although there are specific $y$ values for which $\left<SAT,y\right>$ is hard, there are also an ``easy" $y$'s. Estimating the difficulty of a general SAT problem is known to be hard, therefore, such a system can not verify whether a mining objective is allowed.

\begin{theorem}[Secure Energy Efficient User Uploaded REMS]\label{theorem::UserUploadedSecureREMS}
    Let $\system = ( \objectives,\setup, \eval, \verify)$ be a secure-energy-efficient-user-uploaded REMS. Then for any function $f\in\objectives$, for every $y\in\{0,1\}^{D}$ it must hold that $\comp{\left< f, y \right>}$ is known. \footnote{The difference between this claim and Claim~\ref{claim::weightsOfMiningObjectives} is that here we also quantify over every $y\in\{0,1\}^{D}$, whereas in the other claim we quantify only over $x$'s}
\end{theorem}
 
\begin{proof}
    From Claim~\ref{claim::weightsOfMiningObjectives}, we get that every mining objectives $\left<f,y\right>\in \objectives$ must have the property that $C_{\left<f,y\right>}$ is known to $\system$. 
    Assume towards a contradiction that there exists some $z\in\{0,1\}^{256}$ for which the amount $C_{\left<f,z\right>}$ is unknown to $\system$ (this of course means that $\left<f,z\right>\notin\objectives$.
    Since we assumed that the system allows for user-uploaded problems, we have that an attacker $A$ can upload $\left<f,z\right>$ to $\objectives$. Therefore $\system$ must know (or be able to compute) the amount $C_{\left<f,z\right>}$, which is a contradiction. 
\end{proof}

In the following theorem  we restrict the discussion to the case where the proofs generated by $\eval$ are the output of each objective. In this case, we get that $\objectives$ may contain only preimage resistance functions.

\begin{theorem}[Pre-Image Resistance Property]\label{Theorem::PreImage}
    Let an REMS $\mathcal{S} = (\objectives, \setup, \eval, \verify)$ be secure-energy-efficient-user-uploaded. Assume in addition that the proof generated by $\eval$ at index $i\in [|\objectives|]$ given input $x$ is exactly $f_i(x)$. Then it holds that the family $\objectives$ contains only preimage resistant functions.
\end{theorem}

\begin{proof}
    We need to show that every $f$ such that $\left<f,y\right>\in\objectives$ is pre-image resistant. 
    This means that given $\left<f,y\right>$ for any polynomial time algorithm $ALG$, and every $z \in \{0,1\}^{256}$ it holds that $Pr\left[ALG \textrm{ finds } x \textrm{ such that } f(x)=z\right]$ is negligible in the size of $2^{|f,x,z|}$.
    
    From claim~\ref{claim::energyEfficientProofsNecessarry} we have that $\pi$ which is returned by $\eval$ when running on $x$ must contain proofs of attempts at checking whether $f(x) = y$ for mining objectives $\left<f,y\right>\in\objectives$.
    In addition, from the assumption of this claim we have that given $x,f$,   the proof is exactly $\pi_{\left<f,y\right>} = f(x)$.
    
    Assume towards a contradiction that there exists $\left<f,y\right>\in \objectives$ which is not preimage resistant. There exists some polynomial time algorithm $ALG$ which, for a given input $z \in \{0,1\}^{256}$ can generate $x$'s such that $Pr\left[f(x)=z\right]$ for some mining objective $\left<f,y\right>$, is not negligible. This means that $ALG$ is also a polytime algorithm for generating proofs of attempting to solve $\left<f,y\right>$. This means that $ALG$ is a faster algorithm than $\eval$ for producing the proof $\pi$, which is a contradiction to the optimality of $\eval$
    
\end{proof}

To give an intuition to the above, we can draft the steps as: given a mining objective $\left<f,y\right>\in \objectives$, it holds that 
\begin{align*}
    &Pr[\verify_{\left<f,y\right>}(\objectives, ek, ALG(z), z, z) = \mathbb{T}] = Pr[\eval_{\left<f,y\right>}(\objectives, ek, ALG(z)) = z] = Pr[f(ALG(z)) = z]
\end{align*} \big(the first equality is because $\verify(\objectives, ek, x, y, z)$ is actually the function $\eval_{\left<f,y\right>}(\objectives, ek, x) == f(x)$ for all $\left<f,y\right>$\big). Thus if $f$ is not preimage resistant, then $\verify$ is inversable.

From all that we have shown above, it stems that when designing a proof of useful work system, the designer should decide on whether block creation depends on solving mining objectives.

If mining a block does depend on solving mining objectives, then from Claim~\ref{claim::weightsOfMiningObjectives} we get that for any $\left<f',y'\right>\in\objectives$, $$Pr\left(x \textrm{ results in a block }| f(x) = y\right)  = Pr\left(x \textrm{ results in a block }| f'(x) = y'\right)\cdot \frac{\comp{\left<f,y\right>}}{\comp{\left<f',y'\right>}}$$ And that the system should know how to adjust the odds of block creation according to the amounts $\comp{\left<f,y\right>}$ for all mining objectives that a user may upload to the system. For instance, if the system allows to upload both instances of $\textrm{SHA}256$ and $\textrm{MD}5$, then the amount of computations needed for a pair $\left<\textrm{SHA}256,y\right>$ and $\left<\textrm{MD}5,z\right>$ should be known for any $y,z\in\{0,1\}^{256}$. This is a very harsh restriction for a system designer. One way in which it can be enforced is as in \cite{zhang2017rem} by counting the number of CPU operation based on the additional assumption of trusted hardware.
    
If solving mining objectives does not affect the probability to mine a block (meaning that block creation is independent of solving mining objectives), then  finding a solution to a mining objectives is a byproduct of mining, but not the objective. In this case, we must validate that $\eval$ produces verifiable proofs. In Appendix~\ref{app::implementation} we provide an example of how to implement such system.

\subsection{Explicit Double Spend Attack when not all instances in $\objectives$ are "Every Case Hard"}

An example of a double spend attack against the blockchain system presented in \cite{ball2017proofs} where the mining objectives are allowed to be "hard on average". An attacker uploads an objective $\left<f,y\right>$ which is easier than the average case. for the sake of this example we will use a SAT problem $\Phi = \left(x_1\vee x_2\vee x_3\right)\wedge \ldots \wedge \left(x_{k-2}\vee x_{k-1}\vee x_k\right)$ (where $x_1,\ldots, x_k$ are literals), however an example can be generated for any problem which has an easy\textbackslash hard instance. Assume that the attacker knows that the first clause is satisfied under any assignment. Therefore the attacker has an extra bit of information compared to anyone else on the problem. This means that the attacker has an extra bit of information in the verification phase for every attempt of any $z\in\{0,1\}^{256}$. She simply doesn't have to check if the first clause is satisfied, while all other users do. This means that the attacker increases her relative computational power, enabling her to double spend without a majority of the computational power.

\section{Conclusions and Future Work}\label{sec::conclusions}
In this paper we formally defined the property of energy efficiency in PoW systems in the permissionless setting. We used this definition to fully characterize systems in which the mining mechanism operates as an open market between the problem uploaders (the consumers) and the miners (the producers). We formalize such systems using three properties: (1) Security against double spends by a minority attacker (2) Energy efficiency and (3) user uploaded problems. Using this formulation, we showed a negative result for different types of Proofs of Useful Work concept. 
In Appendix~\ref{app::implementation} we also show an explicit construction of a system that holds all the three properties. 

A natural question is to extend this analysis to alternatives to PoW such as proof of space \cite{dziembowski2015proofs} and proof of stake \cite{king2012ppcoin}. In the case of proofs of space, the question is easily translated into whether can we use proofs of space for storing data in a way which avoids unnecessary data duplication, while making sure the data stored is data that real users are willing to pay to store, all the while being safe against double spends. 
In the case of proofs of stake, the analogy is less natural and has to do with measuring the economical loss of storing money in escrow as compared to keeping it in circulation. We believe that this question is more delicate and is of interest.

In our model we assumed that all blocks must have equal weight. We believe that this can be generalized to a setting in which the weight of a block may vary across blocks. In this work, the system might need to adjust the weight of the block according to the mining objective which resulted in the block creation. This is left as a direction for future work.

One final note is that we assumed that the amount $D_f$ is fixed per function $f$ and that every execution of $\eval$ takes the same amount of computational resources. This clearly need not be the case in general. An workaround could be derived from \cite{zhang2017rem}, where trusted hardware is used to verify the exact amount of work that went into the mining. 
If a system of this type can be designed in a secure way, we believe that this would be of interest.

\bibliographystyle{named.bst}
\bibliography{bib}

\appendix

\section{Resource-Efficient Mining is A Special Case of our Solution}\label{app::REMspecialCase}
In their paper ``Resource-Efficient Mining for Blockchains'' \cite{zhang2017rem}, the authors suggest that miners use special hardware called ``Intel SGX''.
This hardware can provide secure instructions counting, and therefore provide a proof of the invested computational resources that a miner put into the mining process. 

We consider this paper as a special case implementation of our guidelins. As in our protocol, mining a block is independent of mining objectives, which proofs that their system meets the conditions imposed by Claim~\ref{claim::weightsOfMiningObjectives}. 
Their experiments shows the the "overhead" of their protocol is around $5.8\% \sim 14.4\%$ which states that Claim~\ref{claim::evalOptimalProver} is true with $\epsilon \approx 0.13$.
The other claims in our paper hold directly from the design of the secure hardware. Note that they do not have to demand that the mining objectives be preimage resistance functions because the conditions of Theorem~\ref{Theorem::PreImage} do not holds; They build the proofs with the special hardware rather then the output of the mining objectives.

We consider this protocol and ours as two different approaches to implement the idea of REMS that we presented in Section~\ref{sec::modelAnalysis}. On the one hand, Resource-Efficient Mining enforce specific hardware, thus can solve a wider family of mining objectives and be fully dynamic through the lifetime of the system. On the other hand, our protocol does not demand a specific type of hardware, therefore it increases the accessibility for new miners (lower entrance investment) which increases the security of the network.

\section{Security Proofs}\label{app::Security}
This appendix contains proofs to all claims from Section~\ref{subsec::security} and known security definitions.

\begin{claim*}[\ref{claim::evalOptimal} $\eval$ is optimal]
Let $\mathcal{S}=( \objectives, \setup, \eval, \verify)$ be a \textbf{secure} REMS.
Then there does not exist any algorithm $\eval'$ that is more efficient than $\eval$ such that $\mathcal{S'}=( \objectives, \setup, \eval', \verify)$ is sound. \footnote{This condition means that $\eval$ is the optimal algorithm for the computational task which mining is based on.}
\end{claim*}

\begin{proof}
    Assume towards a contradiction that there exist $\eval'\neq \eval$ which is more efficient than $\eval$. Assume w.l.o,g that the execution of $\eval'$ is more efficient than that of $\eval$ by a factor of $\alpha> 1$
    Assume that an attacker uses $\eval'$ instead of $\eval$, while all other users use $\eval$.
    Then it holds that the portion of block awarded to the attacker in expectation is:
    
    $$\frac{\textrm{\# executions of attacker using }\eval'}{\textrm{\# of executions of }\eval + \textrm{\# executions of }\eval' } =$$
    $$\frac{\alpha\cdot\textrm{\# executions of attacker using }\eval}{\textrm{\# of executions of }\eval + \textrm{\# executions of }\eval' } >$$
    $$\frac{\textrm{\# executions of attacker if they used  }\eval}{\textrm{\# of executions of } \eval \textrm{ if everyone used }\eval}$$
    Which is an honest miners' probability of mining a block.
    So an attacker increases the speed at which he\textbackslash she mines a block as compared to the honest network, which is a contradiction to the notion of security defined in~\ref{def::Security}.
\end{proof}

\begin{claim*}[\ref{claim::weightsOfMiningObjectives}]
For any $\left<f',y'\right>\in\objectives$ it holds that:
$$Pr\left(x \textrm{ results in a block }| f(x) = y\right)  = Pr\left(x \textrm{ results in a block }| f'(x) = y'\right)\cdot \frac{\comp{\left<f,y\right>}}{\comp{\left<f',y'\right>}}$$

Where $\comp{\left<f,y\right>}$ is the expected amount of computational power required for Finding a solution $x$ that satisfies $\left<f,y\right>$.
\end{claim*}

\begin{proof}
    Assume towards a contradiction that there exists a mining protocol which is secure that does not mandate the described property. If there exists some $\left<f,y\right>\in \objectives$ such that solving $\left<f,y\right>$ increases the chances of mining a block dis-proportionally to the relative computational power required to compute $\left<f,y\right>$, an attacker may choose to focus on solving $\left<f,y\right>$ instead of using $\eval$, and then run $\eval$ only on the solutions that they discovered for $\left<f,y\right>$. This way, the attacker is in fact utilizing a more efficient algorithm than $\eval$ for mining a block, in contradiction to Claim~\ref{claim::evalOptimal} where we prove that $\eval$ is optimal.
\end{proof}

\section{Energy Efficiency Proofs}\label{app::energyEfficiency}

This appendix contains proofs to all claims from Section~\ref{subsec::energyEfficient}.

\begin{claim*}[\ref{claim::energyEfficientProofsNecessarry}]
    Let $\mathcal{S}=( \objectives, \setup, \eval, \verify)$ be an \textbf{energy efficient} REMS. Then miners must supply proofs of attempting to solve mining objectives from $\objectives$ as part of the mining process. 
    Furthermore, the computational difficulty of computing $\eval$ is lower than the difficulty of 
    finding a proof $\pi$ for which $\verify$ evaluates to $\mathbb{T}$
\end{claim*}

To prove Claim~\ref{claim::energyEfficientProofsNecessarry}  we prove Claims~\ref{claim::proofsAreNecessary}, \ref{claim::uninteresting} and~\ref{claim::invert_verify_is_harder_than_eval}, from which Claim~\ref{claim::energyEfficientProofsNecessarry} follows.

\begin{claim}\label{claim::proofsAreNecessary}
    Let $\mathcal{S}=( \objectives, \setup, \eval, \verify)$ be an \textbf{energy efficient} REMS. Then miners must supply proofs of attempting to solve mining objectives from $\objectives$ as part of the mining process.
\end{claim}
\begin{proof}
    Assume towards a contradiction that an REMS $\mathcal{S}=( \objectives, \setup, \eval, \verify)$ is an \textbf{energy efficient} in which miners do not need to prove that they attempted to solve the mining objectives from $\objectives$. 
    
    We look at the operation of $\eval$. We divide into cases:
    If $\eval$ does not perform checks whether $f(x) = y$ for some $x\in\{0,1\}^{256},\left<f,y\right>\in \objectives$, then we have that $\mathcal{S}$ is not energy efficient. therefore we can assume that $\eval$ does perform these evaluations. 
    Since we assumed that the output of $\eval$ does not contain proofs of attempts at solving (evaluating) the individual objectives $\left<f,y\right>$ along the way,we can consider the following algorithm: $\eval'$ operates in the same way $\eval$ does, but every time $\eval$ checks whether $f(x) = y$ for some $x\in\{0,1\}^{256},\left<f,y\right>\in \objectives$, $\eval'$ outputs $\mathbb{F}$. This makes $\eval'$ faster than $\eval$, in contradiction to the optimality of $\eval$ which was proved in Claim~\ref{claim::evalOptimal}.

\end{proof}

We remind the reader that Claim~\ref{claim::proofsAreNecessary} in fact did not require that the REMS be secure, which differentiates it from all other claims in this section.

\begin{claim}\label{claim::uninteresting}
    In an \textbf{$\epsilon$-energy-efficient REMS}, the amount of computational resources that goes into computing $\eval$ is at most  $$\frac{1}{1-\epsilon} \cdot \sum_{\left<f,y\right> \in \objectives} D_f$$ 
\end{claim}

\begin{proof}
    Immediate from the definition of energy efficient:
    Let $C_T$ be the amount of computational resources that goes into computing $\eval$.
    From $\epsilon$-energy efficiency we have that:
    $$\frac{\Sigma_{\left<f,y\right>\in\objectives} D_f}{C_T } > 1-\epsilon$$ 
    I.e. 
    $$C_T < \frac{\Sigma_{\left<f,y\right>\in\objectives} D_f}{ 1 - \epsilon }$$

\end{proof}

The next claim expands our definition of soundness to incorporate the new addition of proofs that is necessary to encompass our requirements of energy efficiency. The new addition is that the probability of defeating $\verify$ will now also have to be negligible in the size of $\pi$ (and not only in the size of $x$ and $ek$). Until now, $\verify$ just checked that $y$ is indeed the output of  $\eval(\objectives, ek,x)$, but now we also want that verify will examine the proof $\pi$.

\begin{claim} \label{claim::invert_verify_is_harder_than_eval}
    Let $\mathcal{S}=( \objectives, \setup, \eval, \verify)$ be a \textbf{secure-energy-efficient} REMS. 
    Let $D(\verify,ek,x,y)$ be the difficulty of the optimal algorithm for generating $\pi\in\{0,1\}^{256}$ such that $\verify(ek,x,y,\pi) = \mathbb{T}$. Then we have that it must hold that  $\sum_{\left<f,y\right> \in \objectives} D_f \le D(\verify,ek,x,y)$ for all $x,y \in \{0,1\}^{256}$.
\end{claim}

\begin{proof}
    Assume towards a contradiction that $\sum_{\left<f,y\right> \in \objectives} D_f > D(\verify,ek,x,y)$ for some $x,y \in \{0,1\}^{256}$. Then an adversary may choose to invest resources into his algorithm, $\eval'$ which is inverting $\verify$ since it is easier. But from the energy efficiency, the computational resources that are needed to compute $\eval$ are at least $\sum_{\left<f,y\right> \in \objectives} D_f$ (solving the mining objectives). I.e. The attacker found an algorithm which is more efficient than $\eval$ to produce proofs, contradiction to Claim~\ref{claim::evalOptimal}.
\end{proof}

From the above three claims we conclude that Claim~\ref{claim::energyEfficientProofsNecessarry} holds.

\begin{claim*}[\ref{claim::evalOptimalProver}]
    Let $\mathcal{S}=( \objectives, \setup, \eval, \verify)$ be a \textbf{secure and energy efficient} REMS. Then $\eval$ is the optimal algorithm for generating the proof $\pi$ given setup $ek$.
\end{claim*}

\begin{proof}\label{proof::evalOptimalProver}
    The system is energy efficient, therefore $\eval$ compute possible solutions to the mining objectives in $\objectives$ and produce proofs.
    
    Assume toward contradiction that there exists more efficient algorithm $\eval' \neq \eval$ s.t. $\eval'$ that can generate valid proofs. If an adversary has access to $\eval'$ then he\textbackslash she can divert computational power away from solving mining objectives in $\objectives$ (since she can generate a proof $\pi$ without trying to solve the mining objectives --- unlike the other miners). This means that $\eval$ is not optimal, in contradiction to Claim~\ref{claim::evalOptimal}.
\end{proof}

\begin{claim*}[\ref{claim::findProofMeansSolveObjecives}] 
     Let $\mathcal{S}=( \objectives, \setup, \eval, \verify)$ be a 
     
     \textbf{secure-energy-efficient} REMS.
    Let $\objectives = \{\left<f_1,z_1,\ldots,\left<f_n,z_n\right>\right>\}$ then for any
    
    $\left<f_i,y_i\right>,\left<f_j,y_j\right>\in\objectives$ 
    
    and for all $x=x_1x_2\ldots x_n \in\{0,1\}^{256}$, for any algorithm used to generate proofs it holds that:
    \begin{align*}
        &Pr\bigg(\textrm{Find $y,\pi$ s.t.: } \verify(\objectives, ek,x,y,\pi) =\mathbb{T} |
        \eval(\objectives,ek,x)_{\big<f_i, z_i\big>} = \mathbb{T}\bigg)  = 
    \\
        \frac{\comp{\left<f_i,z_i\right>}}{\comp{\left<f_j,z_j\right>}} \cdot &Pr\bigg(\textrm{Find $y,\pi$ s.t.: } \verify(\objectives, ek,x,y,\pi) = \mathbb{T}|
        \eval(\objectives,ek,x)_{\big<f_j, z_j\big>} = \mathbb{T}\bigg)
    \end{align*}
    Where $y$ is the output of $\eval(\objectives,ek,x)$.
\end{claim*}

\begin{proof}\label{proof::findProofMeansSolveObjecives}
    Assume towards a contradiction that there is some $\left<f,z\right>\in\objectives$ such that knowing the solution to $\left<f,z\right>$ increases the odds of finding $\pi$ to satisfy $\verify$ by more than the relative difficulty of solving $\left<f,z\right>$. Then an attacker can choose to invest resources in solving $\left<f,z\right>$, and then in finding such $\pi$ instead of solving all other mining objectives in $\objectives\setminus \{\left<f,z\right>\}$ which contradicts the fact that that $\mathcal{S}$ is an energy-efficient REMS.
\end{proof}

\section{Implementation of Secure Energy Efficient User Uploaded REMS}\label{app::implementation}

In this section we formally define our suggested protocol. Our model strongly corresponds to the original Bitcoin protocol, and as such, any property that has not been specifically mentioned can be assumed to be untouched and remain loyal to the Bitcoin protocol. Our construction does deviate from the Bitcoin protocol in some aspects, namely the block creation rule and a new type of transaction.

\subsection{Bitcoin Mining Protocol}
In Bitcoin, blocks are created in the following way: each miner guesses random strings. For each string $r$, the miner calculates a binary string $SHA2(SHA2(H_{ek}\circ r))$, Where $H_{ek}$ is the header of the block that the miner tries to mine, which contains the previous block hash (which is unpredictable, thus can be considered as $ek$), its address, a commitment to the transactions the block contains and more. A miner gets to mine a new block if it holds that the output of the computation is smaller than some global parameter $D$ (which is referred to as the difficulty parameter). This process is what is called "Bitcoin Mining". 

\subsection{Modified Protocol - A High level Description}
We begin with a high-level description of our \textbf{secure-energy-efficient-user-uploaded REMS}  and in the following sections we formally describe how everything is realized and prove  correctness.

In order to describe the protocol more simply, we consider only one type of problems - trimmed output of $SHA256^m$ for $m \in \mathbb{N}$. We describe how this protocol can be generalized in section \ref{sec::generalize_problem_types}.

In our system, users can upload mining objectives to the system using a new type of transaction. The transaction holds the prize that the user offers in exchange to a solution to his mining objective. The mining objective itself is an output of a $SHA256$ to which they need the corresponding input. 

Uploaded mining objectives are partitioned into ''active" and ''non-active" mining objectives, according to whether they were solved. The mining process works as follows: Miners choose a subset $S$ of a fixed size of the active mining objectives. After committing to this subset, they start guessing binary strings and if they succeed a block is mined.

The commitment is done by writing a Merkle tree of the set $S$ in the header of the block being attempted (similar to what happens in the coin-base transaction in Bitcoin today).

Our solution ensures that a miner will work on all mining objectives they committed to by pipelining the (untrimmed) result of one mining objective as the input to the next. Only the output of the last mining objective in $S$ might generate a block.

The miners, as in Bitcoin, generate a seed that is concatenated to the block header that they try to mine, and use the result as the input to the first mining objective. If a miner finds a solution to one of the mining objectives while attempting to mine a block, she will publish the seed as a new transaction, and collect the prize to the mining objective. A block was mined only if the final output meets the difficulty requirements. 

Let us formally define a class of mining objectives, that we use later to describe the mining objectives that our system will be able to solve.
\begin{definition} SHA256 trimmed output mining objective
    is a mining objective with the form $(m, f, t, y)$ where $m, f, t \in \mathbb{N}, s \in \{0,1\}^*$. In this mining objective, the goal is to find a $x \in \{0,1\}^*$ such that $SHA256(x)^m[f:t] = y$, where $[f:t]$ means to take the bits from index $f$ to index $t$.
\end{definition}
An example mining objective is upper bounding SHA256's output by demanding that the first $D$ bits should equals $00\cdots0$ (which is close to Bitcoin's mining objective for $m=1, f=0, t=D$).

\subsection{Modified Protocol - Formal Description}
The system $\system = ( \objectives,\setup, \eval, \verify)$ contains the following elements:

\subsubsection{Mining Objectives in $\objectives$}\label{oursys::objectives}
User-uploaded mining objectives must be of the form $\left<f,y\right>$ where $f$ is a function and $y$ is a well formatted output of $f$. We define $\objectives$  be limited to contain only trimmed outputs of $SHA256^m$ for $m \in \mathbb{N}$. We later discuss how this family can be expanded slightly, while still keeping in line with the results from the previous section.

Users may upload mining objectives through a new special type of transaction, which will contain a description of the mining objective, alongside a deposit which can be withdrawn trough supplying a valid solution to the mining objective\footnote{The solution has to be well formatted in the sense that it has to contain the header information of the relevant block at the time of solution. This is important in order to proof compliance with the requirement of unpredictability describes in the firs section.}. When a miner finds a solution to a mining objective, she can publish a transaction with the solution. The solution contains the solver's public key within it as the recipient of the deposited prize. So in order to hijack the solution a miner must be able to find collisions in SHA256. A schematic illustration of this mechanism can be found in Figure~\ref{figure::block_creation_mechanism}.

In addition we limit $\objectives$ to be the set of ``Active mining objectives'', defined in the follwoing way: 
Given a user-uploaded mining objective, it is considered \textbf{active} as long as it complies with the two following requirements:
\begin{itemize}
    \item It has not yet been solved in previous blocks.
    \item It's solution isn't a part of the block's transactions.
\end{itemize}

\subsubsection{$\setup()\rightarrow ek$}
$\setup()\rightarrow ek$ is simply the hash (SHA256) of the latest block header in the system (the last leaf on the longest chain).

\subsubsection{ $\eval(\objectives, ek, x) \rightarrow (y, \pi)$}
Given a set of the active mining objectives: $\objectives = \left<f_1,y_1\right>,\ldots,\left<f_{|\objectives|},y_{|\objectives|}\right>$.
A miner first chooses a subset $S \subseteq \objectives$ of size k. The miner will query $\eval$ using the input parameters $S, ek, x$ \footnote{If there aren't enough mining objectives in $\objectives$, then the miner most add mining objectives to $S$ from a list of mining objectives accepted by the system}. The miner will then calculate $s_0 = H$, where $H$ holds the information on the block that he's trying to mine (including the previous block hash $ek$, his identity, Merkle root of $S$ and Merkle root of the transactions that he includes in the block) 
and then for every active user-uploaded mining objective, the miner checks if the assignment $f(s_{i-1})=y_i$. If not, the miner sets $s_i = f_i(s_{i-1})$ and keeps going. \footnote{If the trimming of $f_{i-1}$ is to strong, then this may degenrate the outputs space of $f_i(s_{i-1})$- this is why we will always use the output before the trimming.}. If at any point $l\in[{|S|}]$ it holds that $f(s_{l-1})=y_l$, the miner may publish a transaction with the proof\footnote{The proof will be $s_0$, and this way the miner is safe from anyone hijacking the solution, since $B$ contains the miners' information} and collect the fee offered by the mining objective-uploader. 
Finally, a block is mined if $s_{|S|} \le D$ (where $D$ is the difficulty parameter). A schematic illustration of this mechanism can be found in Figure~\ref{figure::mining_mechanism}.
The output $y$ is generated in the following way:
\begin{itemize}
    \item $y[0] = \mathbb{T}$ if  $s_{|\objectives|} \le D$ (where $D$ is the difficulty parameter on the system). Otherwise, $y[0] = \mathbb{F}$.
    \item For all $i\in[|\objectives|]$, $y[i] = \mathbb{T}$ if it holds that $f(s_{i-1})=y_i$. Otherwise,  $y[i] = \mathbb{F}$. 
\end{itemize}

The proof $\pi$ is simply the output of the last objective. I.e. $\pi = f_{|\objectives|}(s_{|\objectives|-1})$ 

A very important comment is that $k$ should be equal to the number of different functions (before the trimming) that are in $S$. If $S$ contains multiple mining objectives that are a different trimming of the same output, then the miner needs to calculate the output once, and check all of the possible $y_i$'s against this output.
We assume the overhead of these checks, given the output of the function, is negligible (and can be computed in parallel) when compared to the execution  of the function\footnote{We can remove this assumption and allow only a single occurrence of each $f$ in $S$.}.

Otherwise $\eval$ is not the optimal algorithm for solving the group of the mining objectives (altogether), which is a contradiction to the definition of energy efficient.

\begin{figure}
 \centering 
 \includegraphics[width=0.4\textwidth]{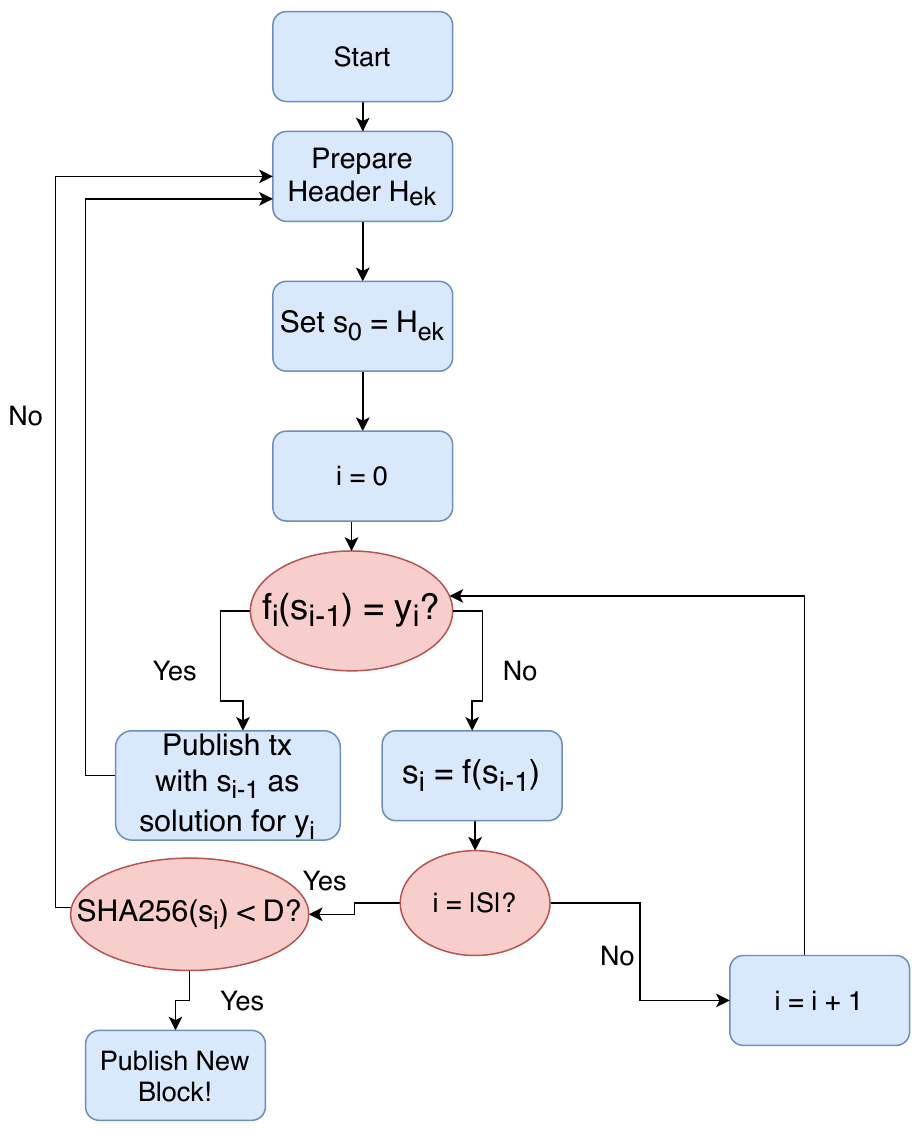}
 \caption{Modified block creation mechanism to allow solving user uploaded preimage resistance functions as part of the mining process.}
 \label{figure::block_creation_mechanism}
\end{figure}

\subsubsection{ $\verify(\objectives, ek, x, y, \pi)\rightarrow \{ \mathbb{T}, \mathbb{F} \}$}
$\verify$ works as expected: Checks whether $\pi$ equals $f_{|S|}(s_{|S|-1})$, and that $y[0] = \mathbb{T}$ (i.e. $s_{|S|} \le D$). 

\subsubsection{Prize Collection}
If a miner finds a solution to a mining objective, she publishes a special transaction with the solution. Since the seed to the solution contains the public key of the miner, everyone can verify that the solution is correct and that she is the legal recipient of the prize. Miners are incentivized to include this transaction in their newly created block because of the transaction's fee, just like any other transaction.

\subsection{Liveness}

\begin{figure*}
 \centering 
 \includegraphics[width=0.8\textwidth]{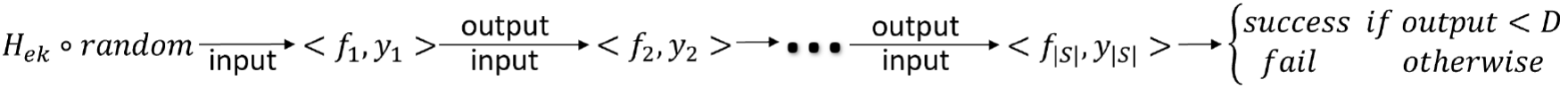}
 \caption{Our modified mining mechanism.}
 \label{figure::mining_mechanism}
\end{figure*}

We notice that the probability of mining a block is the same compared to the Bitcoin protocol, so our guarantees of liveness stems from that of the original Bitcoin protocol. Meaning that the system continues to create blocks even if no mining objectives are uploaded by users (though in this case there does exist the same waste of energy as in Bitcoin).

Moreover, we keep the incentives for the honest miners to keep mining blocks because they still get the block reward and the transactions fees. Therefore any honest transaction will eventually end up deep enough in an honest chain, assuming there's an honest majority.

The only thing to make sure is that the new special transactions will be included. We claim that they will, since hijacking the solution is computationally infeasible (for an adversary whom cannot reverse preimage resistant functions). So miners gain nothing be ignoring such transactions, and are incentivised to include them via transaction fees.


\subsection{Generalizing and Restricting the Family of allowed mining objectives} \label{sec::generalize_problem_types}

In subsection~\ref{oursys::objectives} We limited the discussion to trimmed $SHA256^m$ problems to simplify the model. This section will discuss a possible generalization, in which we broaden the allowed set of mining objectives. We show that although we can generalize our protocol and keep it secure, the allowed mining objectives and protocol still must have restrictions.

From Claim~\ref{claim::weightsOfMiningObjectives} and the fact that mining a block and solving a mining objective is independent, we need to enforce that each attempt to mine a block has the same computational demands. Each attempt to mine a block is actually executing all the mining objectives in the chosen group of mining objectives (that was denoted by $S$), therefore we have to constraint the possible groups. 

One option to do so is by defining a "score" to each type of mining objective. Then, the protocol can enforce ``fixed score" for $S$ and thus control the computational resources of each execution of $\eval$ (each attempt to mine a block). Note that this implicitly holds in the above suggestion because we demanded a ``fixed size" if $S$, and there is only a single type of allowed mining objective.

We offer the following example to guide the readers intuition - Suppose the system allowed two families of functions as mining objectives -  trimmed outputs of SHA$256$ and trimmed outputs of MD$5$ (instead of allowing only trimmed SHA$256$). $\eval$ will be designed as follows - always run SHA$256$ for $k_1$ times and then MD$5$ for $k_2$. Such a system meets all the formal requirements for security, user uploaded mining objectives and energy efficiency (since these are questions uploaded by users whom are willing to pay for the output of the computation). This can be further extended to include other one-way functions, as long as their proportions remain controlled.

\end{document}